\documentclass[a4paper,onecolumn,11pt,accepted=2025-10-29]{quantumarticle}
\pdfoutput=1
\pdfoutput=1
\usepackage{amsmath}
\usepackage{amsthm}
\usepackage{mathrsfs,mathbbol}
\usepackage{amssymb}
\usepackage{graphicx}
\usepackage{caption}
\usepackage{subcaption}
\usepackage[table]{xcolor}
\usepackage{enumerate}
\usepackage{xcolor}
\usepackage{hyperref}
\usepackage{physics}
\usepackage[T1]{fontenc}
\usepackage[symbol]{footmisc}
\usepackage{algorithm}
\usepackage{algpseudocode}
\usepackage[numbers,sort&compress]{natbib}
\usepackage{todonotes}
\usepackage{cancel}
\usepackage[final]{changes} 
\definechangesauthor[name={Alice}, color=red]{alice}

\newtheorem{theorem}{Theorem}[section]

\newcommand{\bk}[1]{\left|#1\right\rangle}

\begin{document}

\title{Multipartite Entanglement Distribution in Quantum Networks using Subgraph Complementations}

\author{Aniruddha Sen}
\thanks{currently at the Department of Computer Science, University of Texas at Austin, Austin, TX, 78712, USA; Email \href{aniruddhasen@utexas.edu}{aniruddhasen@utexas.edu} }
\affiliation{College of Information and Computer Science, University of Massachusetts Amherst, 140 Governors Dr, Amherst, Massachusetts 01002, USA}
\author{Kenneth Goodenough}
\affiliation{College of Information and Computer Science, University of Massachusetts Amherst, 140 Governors Dr, Amherst, Massachusetts 01002, USA}
\author{Don Towsley}
\affiliation{College of Information and Computer Science, University of Massachusetts Amherst, 140 Governors Dr, Amherst, Massachusetts 01002, USA}

\maketitle

\begin{abstract}
     Quantum networks are important for quantum communication, enabling tasks such as quantum teleportation, quantum key distribution, quantum sensing, and quantum error correction, often utilizing graph states—a specific class of multipartite entangled states that can be represented by graphs. We propose a novel approach for distributing graph states across a quantum network. We show that the distribution of graph states can be characterized by \textit{a system of subgraph complementations}, which we also relate to the minimum rank of the underlying graph and the degree of entanglement quantified by the Schmidt-rank of the quantum state. We analyze resource usage for our algorithm and show that it improves on the number of qubits, bits for classical communication, and EPR pairs utilized, as compared to prior work. In fact, the number of local operations and resource consumption for our approach scales linearly in the number of vertices. This produces a quadratic improvement in completion time for several classes of graph states represented by dense graphs, which translates into an exponential improvement by allowing parallelization of gate operations. This leads to improved fidelities in the presence of noisy operations, as we show through simulation in the presence of noisy operations. We classify common classes of graph states, along with their optimal distribution time using subgraph complementations. We find a sequence of subgraph complementation operations to distribute an arbitrary graph state which we conjecture is close to the optimal sequence \cite{Buchanan2021SubgraphComp}, and establish upper bounds on distribution time along with providing approximate greedy algorithms.
\end{abstract}

\section{Introduction}
The distribution of graph states in quantum networks has been extensively studied \cite{hein2004multiparty, meignant_markham_grosshans_2019,fischer2021distributing}. Creating entangled graph states efficiently is important for the development of quantum networks, having applications in areas such as quantum metrology \cite{quantum_metrology}, quantum error correction \cite{Error-correctingEntanglementTree} and distributed quantum computing \cite{distributedqc}. Graph states also form the basis of the one-way model of computation for quantum computers, which is universal and operates by performing single qubit measurements on an initially created resource graph state \cite{raussendorf2001oneway}. 

In this paper, we first consider a quantum network in ideal conditions with no noisy communication and a limited number of perfect quantum memories. We address the problem of distributing a graph state under such conditions and present an algorithm to do so for a network having an arbitrary topology. We are particularly interested in gaining a better understanding of the limitations in resource consumption (of resources such as EPR pairs and number of qubits) and runtime associated with distributing graph states, and how they can be optimized. Secondly, we show that our algorithm leads to improved fidelities of the graph state in the presence of noisy operations.

\subsection{Related Work}
The distribution of arbitrary graph states in a quantum network has received attention recently. Some of the key parameters for distributing graph states are the number of Bell pairs consumed and the number of time steps required. Two approaches similar to our work (i.e.~a noiseless network with an arbitrary topology) are provided by \cite{meignant_markham_grosshans_2019,fischer2021distributing}. The authors of \cite{meignant_markham_grosshans_2019} do not optimize resource consumption for an arbitrary graph state, but provide a framework for future work on this problem by analyzing resource consumption in a noiseless and lossless network. Their algorithm utilizes $(n-1)$ Bell pairs in one time step to distribute an $n$-qubit GHZ state, and $\frac{n(n-1)}{2}$ Bell pairs in $(n-1)$ time steps to distribute an arbitrary graph state with $n$ vertices \cite{meignant_markham_grosshans_2019}. Another approach using an identical setting \cite{fischer2021distributing} utilizes $\frac{3n^2-2n}{8}$ Bell pairs and takes $(n-1)$ time steps to distribute an arbitrary graph state. Both of these approaches assume that multiple gates can be applied on a single qubit in one time step\added{, since this is negligible in comparison with the time for probabilistic Bell pair generation}. However, unlike these approaches, we consider the gate-based model of computation in which a qubit can only have one gate applied to it in one time step. We also assume that Bell pairs are distributed in advance between some central node and end nodes. This increases the distribution times of \cite{meignant_markham_grosshans_2019} and \cite{fischer2021distributing} by a factor of $n$, so they both require $O(n^2)$ time steps.

Other work \cite{avis_rozpedek_wehner_2023,cuquet2012growth,bugalho2023distributing} considers additional factors such as the probabilistic nature of generating and the time for distributing Bell pairs, accounting for memory decoherence and noisy Bell state measurements. These factors significantly impact the rate and fidelity of distributing a graph state and are of interest to practically implement such algorithms in a physical network.

\added{The algorithm provided by Avis et al. \cite{avis_rozpedek_wehner_2023}, Cuquet, Calsamiglia \cite{cuquet2012growth}, and Fischer, Towsley \cite{fischer2021distributing}, (under the assumption that the physical network has a star topology) improves on the work of \cite{meignant_markham_grosshans_2019} by proposing the following approach, \label{fntalg} which we refer to as the \textit{Factory Node + Teleportation Algorithm} or the FNT algorithm in the rest of the paper.}
\noindent\rule{\linewidth}{2pt}

\noindent\textbf{\added{FNT graph state distribution algorithm}}

\vspace*{-0.7em}
\noindent\rule{\linewidth}{1pt}

\begin{enumerate}[1.]
    \item Prepare the desired $n$-qubit graph state locally at a central node
    \item Distribute Bell pairs between the central node and each of the $n$ end nodes
    \item Teleport each qubit in the central node to each of the end nodes, consuming all Bell pairs in the process
\end{enumerate}

\noindent\rule{\linewidth}{2pt}

This algorithm is also described by Cuquet, Calsamiglia as the \textit{Bipartite A} protocol \cite{cuquet2012growth}; they additionally provide an improved algorithm called \textit{Bipartite B} which uses only $n$ qubits in the central node and fewer gates. However, they do not consider noise from local correction unitaries, which differs from our noise model as described in section \ref{sec:noisy}. We show improved experimental results for fidelity in comparison to both models. These approaches utilize $n$ Bell pairs, $O(|E|) = O(n^2)$ time steps, and require $2n$ or $n$ qubits to be stored in the central node. We improve upon each one of these factors as described in the following section.

Note that the distribution times for some of the above approaches can be improved to $O(n)$ by parallelizing gate operations. We also analyze the effect of incorporating parallelization in our approach and show a comparable speedup for distributing all graph states, and an exponential improvement in distribution time, i.e. $O(\log n)$ time steps, for certain relevant classes of graphs.

\subsection{Summary of Contributions}
We propose a novel graph state distribution algorithm both with and without parallelization and analyze it first for the noiseless case. We then examine the effects of noise through simulations on various graph states. Our algorithm is based on a specific operation on graph states, termed a \emph{subgraph complementation} ($SC$)\cite{Buchanan2021SubgraphComp}, which should not be confused with local complementation. We then show that any graph state can be distributed in the network through a sequence of $SC$ operations. Along with introducing this algorithm, we make the following contributions, improving upon prior work.

\begin{itemize}
    \item We show that our graph state distribution algorithm utilizes $(n-1)$ Bell pairs and requires at most $(n-1)$ qubits at the central node. The number of Bell pairs is clearly optimal and the number of qubits at the central node approximately scales with the average degree of the vertices in the graph.
    \item We show that our algorithm requires $O(n)$ time steps and $O(n)$ bits of classical communication to distribute sparse graphs, as well as several relevant classes of dense graphs such as complete graphs (GHZ states) and complete bipartite graphs, but may take $O(n^2)$ resources for some dense graphs.
    \item By parallelizing gate operations and using auxiliary qubits, we show that the distribution time improves to $O(\log n)$ time steps for the same classes of dense graphs. We obtain a general upper bound of $O(n\log n)$ time steps for arbitrary  dense graphs and $O(n)$ for sparse graphs. 
    \item We establish a relation between the number of $SC$ operations needed to distribute a graph state $\ket{G}$ and the entanglement of $\ket{G}$ (quantified by its average Schmidt-rank over all possible bipartitions).
    \item We also carry out simulations (with and without parallelization) of our algorithm in the presence of gate operations characterized by depolarizing noise, and show that our algorithm almost always increases fidelity compared to other approaches.
\end{itemize}

\section{Background}
Graph states are a class of multipartite entangled states that can be represented by graphs, with each vertex representing a qubit and each edge representing a $CZ$ (controlled-$Z$) gate. A graph state $\bk{G}$ corresponding to the undirected graph $G = (V,E)$ is defined as
\[\bk{G} = \prod_{(a,b) \in E}CZ_{ab}\bk{+}^{|V|}\,, \]
where $CZ_{ab}\bk{+}$ denotes an edge between two vertices $a,b \in V$, each of them initially in the $\bk{+}$ state, with $n=|V|$. Applying a $CZ$ gate twice removes an edge, i.e. $CZ^2 = I$. Note that $CZ$ gates commute with each other and are also symmetric, so edges are undirected.

We utilize the following graph operations corresponding to quantum gates or measurements to modify a graph state in this paper --- edge addition/deletion ($CZ$ gate), local complementation of a vertex ($LC$) and measurement of a qubit in either the $Y$ or $Z$ basis. An $LC$ operation on a vertex complements the induced subgraph formed by the neighborhood of a vertex, whereas a $Z$-basis measurement removes a vertex from the graph \cite{hein2004multiparty}. A $Y$-basis measurement is equivalent, up to single qubit Clifford gates, to an $LC$ operation followed by removal of the vertex \cite{hein2004multiparty}.

We also need an operation denoted as \textit{subgraph complementation} \cite{Buchanan2021SubgraphComp} (or $SC$) on some subset of vertices $V' \subseteq V$ of the graph. Subgraph complementation of $V'$ corresponds to complementing the induced subgraph on vertices $V'$ of the graph $G$. That is, for each pair of vertices $v_i,v_j \in V'$, we either add or remove the corresponding edge $(v_i,v_j)$.

An $SC$ operation is similar to an $LC$ operation since an $LC$ on a particular vertex $v$ is equivalent to an $SC$ carried out on the neighborhood of $v$.  However, $SC$ operations do not correspond to $LC$ operations since they can be on any arbitrary induced subgraph of $G$. The one exception is when the subgraph is the neighborhood of $v$ for some $v \in V$ not in the subgraph. Also, note that $SC$ operations commute.

\section{Subgraph complementation operation}
We describe the subgraph complementation operation for the setting of no noise and deterministic Bell pair generation. This section will be followed by a discussion of $SC$ operations as part of our algorithm to distribute an arbitrary graph state.
\subsection{Single subgraph complementation}
\label{SCsingle}
We first present an algorithm for performing a single $SC$ operation on any subset of nodes $V' \subseteq V$ of the network, requiring $O(|V'|)$ time steps. This will serve as a building block for our graph state distribution algorithm.
 
\added{Roughly speaking, we construct the desired graph state step by step through $SC$ operations that consist of local CZ gates (acting only on $n-1$ qubits as opposed to the $2n,n$ qubits in \cite{avis_rozpedek_wehner_2023,cuquet2012growth}) and local complementations.}

More formally, consider a network with a central node $A$ that contains $k-1$ qubits $\{a_i: 0 \le i < k-1\}$. The network also contains a set of $k-1$ nodes $\{C_i\}_{i=0}^{k-2}$,  each containing one qubit $\{c_i\}_{i=0}^{k-2}$. The qubit $a_0$ will be used as a `central' representative qubit for node $A$. Our setup assumes that $(k-1)$ Bell pairs of the form $(a_i,c_i), i = 0,\ldots, k-2$, are already distributed between the nodes. Note that this defines two network graphs corresponding to the nodes (the physical graph) and qubits (the graph state) respectively. We make no assumptions about the underlying topology or connectivity of the network, since our algorithm only requires Bell pairs to be distributed between an arbitrary central node and all the other nodes. This can be achieved starting from any connected physical graph, though the optimal method to do so may be nontrivial to find. {For the analysis in this paper we focus on a star network topology, or equivalently, we exclude all the resources and contributions to the noise required to connect the end users to the central node. We focus on this setting since a star network topology is considered to be a reasonable architecture that is implementable in the near term.}

The simplest $k$-qubit graph state we can distribute using one $SC$ operation is the $k$-GHZ state or a complete graph state on $k$ qubits (i.e. a $k-$\textit{clique}). The following algorithm describes how to distribute a $k$-clique among all the nodes $V=\{A\}\bigcup\{C_i\}_0^{k-2}$. This involves carrying out one $SC$ operation on $k-1$ nodes on $V'=\{C_i\}$ followed by measurements in $A$.  An example operation for $k=4$ can be found in Fig. \ref{fig:SC4clique}.

\begin{figure}[h]
  \centering
  \includegraphics[scale=0.25]{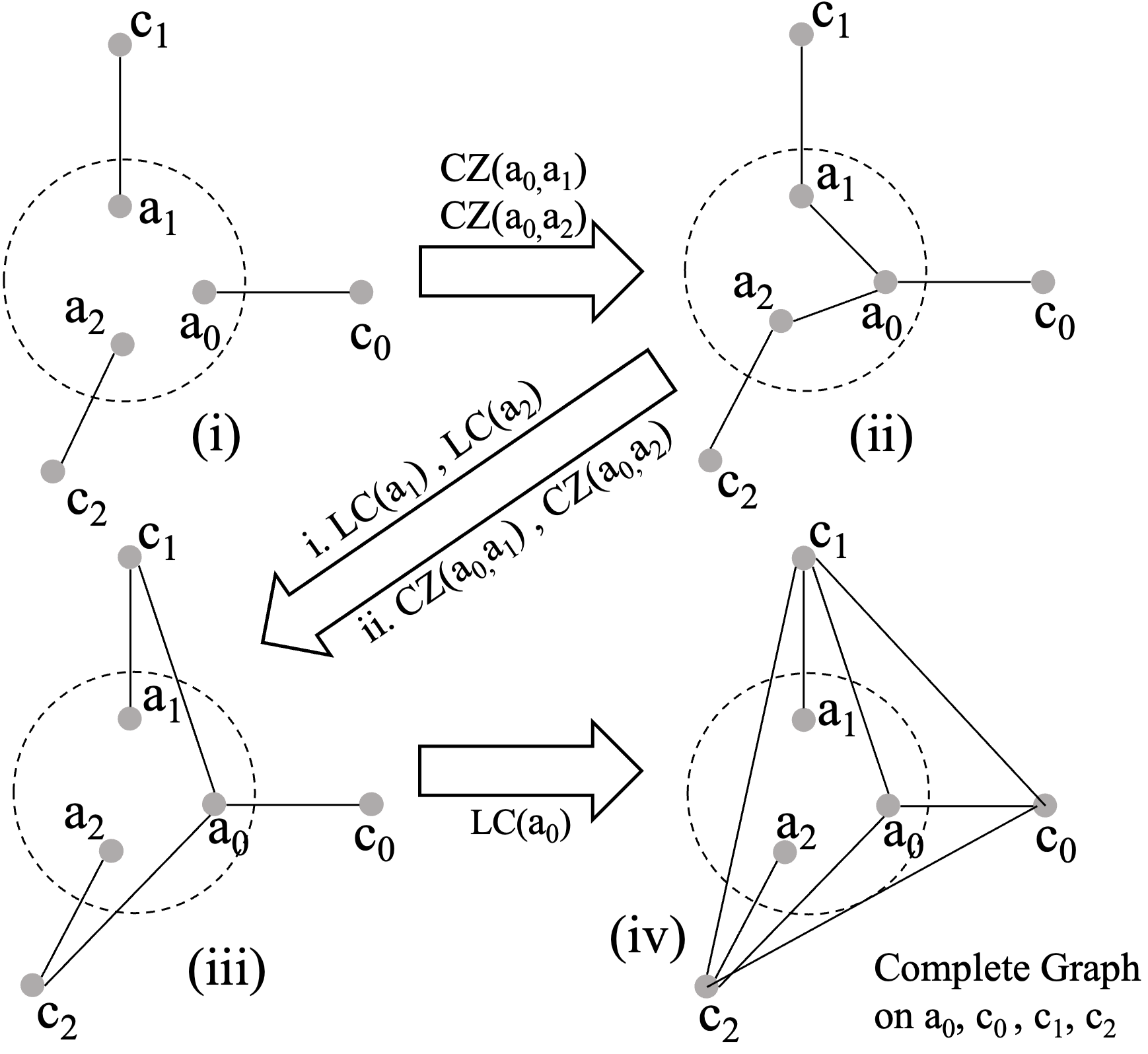}
  \caption{4-clique addition operation. The SC operation is first carried out on $V'=\{C_0,C_1,C_2\}$ containing qubits $\{c_0,c_1,c_2\}$.} If a 4-clique is the final state to be distributed, qubits $a_1$ and $a_2$ need to be measured out in the $Z$-basis at the end. Else, preserve the qubits for further $SC$ operations.
  \begin{picture}(0,0)
    \put(-110,230){A}
  \end{picture}
  \label{fig:SC4clique}
\end{figure}

\noindent\rule{\linewidth}{2pt}

\noindent\textbf{Subgraph complementation operation on $V'$}

\vspace*{-0.7em}
\noindent\rule{\linewidth}{1pt}

\begin{enumerate}[1.]
  \label{SC}
  \item $\forall a_i: i \ne 0 \implies CZ(a_i,a_0)$ --- Connect all qubits in $A$ with an edge to $a_0$. The resulting state corresponds to state (ii) in the figure for 4 qubits.
  \item $\forall a_i: i \ne 0 \implies LC(a_i)$ --- Locally complement each qubit in $A$ except $a_0$. 
  \item $\forall a_i: i \ne 0 \implies CZ(a_i,a_0)$ --- This is a repeat of the first step, except this time the edge between $a_i$ and $a_0$ is deleted. The resulting state corresponds to state (iii).
  \item $LC(a_0)$ --- The resulting state corresponds to state (iv). \added{This complements the induced graph on $V'$, which was initially an empty graph}.
\end{enumerate}
\noindent\rule{\linewidth}{2pt}

If a $k$-clique (corresponding to a $k$-qubit GHZ state) is the final state we want to distribute, then some qubits may need to be measured out in the $Z$-basis at the end (e.g. $a_1$ and $a_2$ in Fig. \ref{fig:SC4clique}). However, if we want to distribute an arbitrary graph state, then it is beneficial to actually preserve some or all connections of the form $(a_i,c_i)$  between qubits for carrying out additional $SC$ operations, in order to avoid sharing a new set of Bell pairs each time. This is described in more detail below.

\subsection{Sequence of subgraph complementations}
In order to carry out another $SC$ operation after obtaining a state similar to state (iv) in Fig. \ref{fig:SC4clique}, we need to restore direct connections $(a_i,c_i)$ between the central node and each of the end nodes, similar to the connections in state (i) in Fig. \ref{fig:SC4clique}. We refer to this as an \textit{edge reset} operation. This simply involves carrying out steps 1--3 in the subgraph complementation operation in reverse order. This is illustrated in Fig. \ref{fig:edge_reset}. The only difference between the new state (v) and state (i) is that the end nodes form a clique. The $SC$ operation can then be carried out on a new set of vertices to add another clique, but since $CZ^2 = I$, the edges of the clique get added modulo 2 (every two additions of the same edge removes it). More specifically, applying the sequence of operations 1--4 on any set of vertices $V' \subseteq V$ corresponds to performing another $SC$ operation on $V'$. This can be used to successively perform several $SC$ operations. 

Finally, after performing some sequence of $SC$ operations, we would like to select which connections (edges) to keep or not to keep between the central node $A$ and the other nodes. Since $A$ is also part of the network, simply perform step 1 followed by step 2 of the $SC$ operation to add any edges between $A$ (resp. $a_i$) and $C_i$ ($c_i$), and measure out all the other qubits. 
\begin{figure}[h]
  \centering
  \includegraphics[scale=0.45]{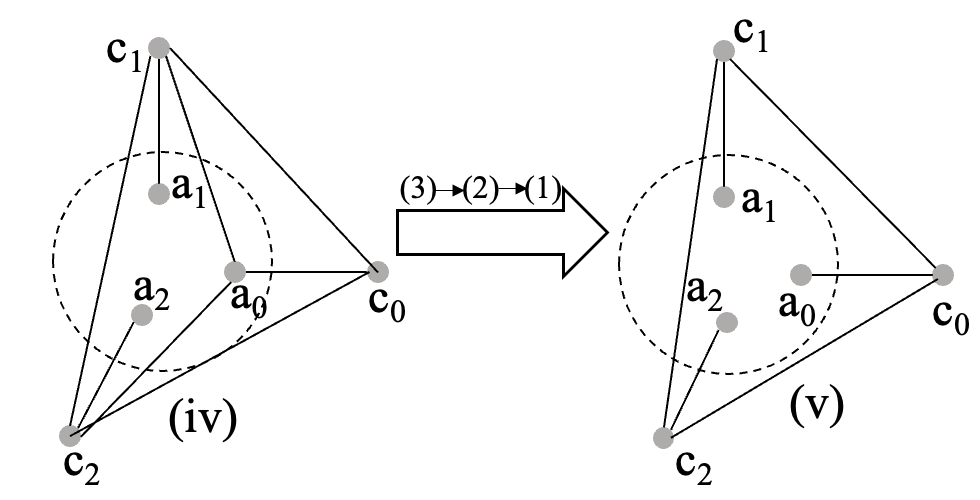}
  \caption{Edge reset operation.}
  \label{fig:edge_reset}
\end{figure}

\subsection{Analysis}
\label{sc_analysis}
In order to carry out an $SC$ operation on \added{$k'=k-1$} vertices in a network that includes the central node, we utilize \added{$k'$ Bell pairs and $k'$ qubits} in the central node. A total of \added{$(2k'-2)$ $CZ$ operations and $k'$ $LC$ operations are required by the SC operation \ref{SC}}. This is followed by up to $k'$ $Z$-measurements, up to $k'-1$ $LCs$ and $CZs$ in conjunction with each other \added{depending on the number of edge resets}. Thus, the number of time steps required to carry out a subgraph complementation followed by an edge reset is bounded by $6k'-4=O(k')$. 

\section{Graph State Distribution Algorithm}
In this section, we describe an algorithm for distributing an arbitrary graph state, which uses $SC$ operations as a primitive. Specifically, we show that any graph state can be generated through a sequence of $SC$ operations, and also analyze the distribution time.

\subsection{Subgraph complementation system}
Any arbitrary graph $G=(V,E)$ with $n$ vertices can be represented by a subgraph complementation system $\mathcal{SC}_{sys} = \{V_1, \ldots, V_d \}$ \cite{Buchanan2021SubgraphComp} where $V_i$ denotes a single $SC$ operation on $V_i\subseteq V$ and $d = |\mathcal{SC}_{sys}|$. Note that $SC$ operations can be carried out in any order, and that this decomposition is not unique. Let the subgraph complementation number $\mathbb{c}_2(G)$ denote the minimum cardinality $d$ (over all possible $\mathcal{SC}_{sys}$ for a graph $G$) of such a set \cite{Buchanan2021SubgraphComp}. For instance, consider the case of building the wheel graph $W_5$ as shown leftmost in Fig. \ref{fig:SCsys_wheel}. For one possible $\mathcal{SC}_{sys}$, all 8 edges (which are 2-cliques) are added one by one, corresponding to one $SC$ operation each, so $d = 8$. 
\begin{figure}[h]
  \centering
  \includegraphics[scale=0.3]{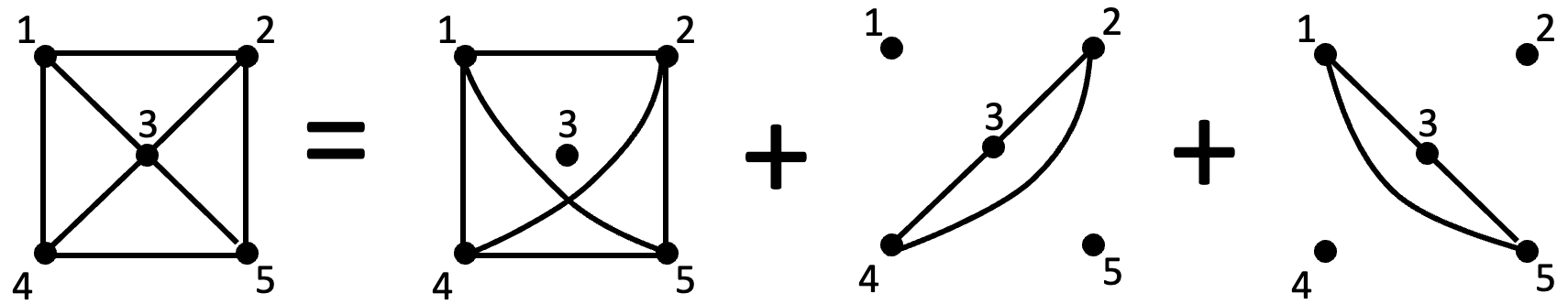}
  \caption{Subgraph complementation system for $W_5$ \cite{Buchanan2021SubgraphComp}.}
  \label{fig:SCsys_wheel}
\end{figure}
In the example shown in the figure, only three $SC$ operations (solid lines) are required and $\mathcal{SC}_{sys} = \{\{1,2,4,5\}, \{2,3,4\}, \{1,3,5\}\}$, where each subset of $\mathcal{SC}_{sys}$ contains the vertices in an $SC$ operation according to the numbering in Fig. \ref{fig:SCsys_wheel},  thus $d = 3$. In fact, we can show that this is the minimum number of $SC$ operations required to build the given graph $W_5$. Thus, $\mathbb{c}_2(W_5) = 3$.

\added{Now, we describe an algorithm for distributing a graph state $\ket{G}$ according to a fixed subgraph complementation system $\mathcal{SC}_{sys}$. We also refer to this as the $SC$ algorithm.}

\noindent\rule{\linewidth}{2pt}

\noindent{\textbf{Graph state distribution algorithm via $\mathcal{SC}_{sys}$}}

\vspace*{-0.7em}
\noindent\rule{\linewidth}{1pt}

\begin{enumerate}[1.]
  \label{SCalg}
  
  \item {\textbf{Initialization.}  
  Given a fixed $\mathcal{SC}_{sys}$ for $G$, start with a fixed set of $n-1$ Bell pairs between an arbitrary central node and all the other nodes in the network.}

  \item {\textbf{Apply subgraph complementation operations.}  
  For each $V_i \in \mathcal{SC}_{sys}$, perform the corresponding $SC$ operation on $V_i$, followed by an edge reset on $V_i$. Each operation complements a $|V_i|$-size clique for the vertices in $V_i$ in $G$. The order of the $|\mathcal{SC}_{sys}|$ operations does not affect the final graph state.}

  \item {\textbf{Last edge reset.}  
  After the last $SC$ operation, add or remove the appropriate edges between the central node and other nodes in $G$ (the first 3 steps of the $SC$ operation \ref{SCsingle} allows us to add any set of edges $(a_0,c_i)$ as in Fig.~\ref{fig:SC4clique}, and we can measure out qubit $a_i$ to remove the edge $(a_0,c_i)$).}

  \item {\textbf{Measure.}  
  Measure out all qubits in the central node that are not part of $\ket{G}$. Now, the state $\ket{G}$ corresponds to the desired graph state for $G$.}

\end{enumerate}

\noindent\rule{\linewidth}{2pt}

\subsection{Time for subgraph complementation}
For our analysis, we assume that Bell pairs are distributed in advance. We consider first the gate-based model of computation with perfectly noiseless quantum channels for communication (see Section \ref{sec:noisy} for a numerical analysis in the presence of noise). Let $T(G,\mathcal{SC}_{sys})$ denote the time required by our algorithm to distribute a graph state $\bk{G}$ according to subgraph complementation system $\mathcal{SC}_{sys}$. Denote by $t(k)$ the time required for a single $SC$ operation on $k$ vertices, $k \in \mathbb{N}$. It follows from the prior analysis section that $t(k) = \gamma k = O(k),$ for some $ \gamma \in \mathbb{R^+}$, since time is proportional to the number of $CZ$ gates performed. This directly leads to an expression for $T(G,\mathcal{SC}_{sys})$ as stated below.

\begin{theorem}
Given a graph $G$ with $n$ vertices and the corresponding system $\mathcal{SC}_{sys}$ with minimal size $\mathbb{c}_2(G)$ such that each $SC$ operation on $V_i$ is performed on $k_i$ vertices, then $T(G,\mathcal{SC}_{sys}) = \mathbb{c}_2(G) \cdot t(\overline{k})\,,$ where $\overline{k} = \frac{1}{d}\sum_{i=1}^{d}k_i$.
\label{GSDtime}
\end{theorem}
\begin{proof}
Given a graph $G$ and an $SC$ system $\mathcal{SC}_{sys}$,
\[T(G,\mathcal{SC}_{sys}) = \sum_{i=1}^{d}t(k_i)  \,,\]
where $k_i$ denotes the number of vertices in $V_i$.
By linearity of $t$, 
\[\overline{k} = \frac{1}{d}\sum_{i=1}^{d}k_i \implies t(\overline{k}) = \frac{1}{d}\sum_{i=1}^{d}t(k_i) = \frac{1}{d}T(G,\mathcal{SC}_{sys}).\]
Also assuming that $d$ is minimum,
\[T(G,\mathcal{SC}_{sys}) = \mathbb{c}_2(G) \cdot t(\overline{k}) \,.\]
\end{proof}

It is possible that $T(G,\mathcal{SC}_{sys})$ can be minimized by some $\mathcal{SC}_{sys}$ that is not of minimum size $\mathbb{c}_2(G)$. However, the case $d = \mathbb{c}_2(G)$ provides an upper bound $T(G,\mathcal{SC}_{sys}) = \mathbb{c}_2(G) \cdot t(\overline{k})$ on distribution time, which we conjecture is a tight bound for most graph states. \added{The circuit depth or the number of local operations to distribute $\ket{G}$ according to $\mathcal{SC}_{sys}$ is bounded by $6nd$ since $\sum_{i=1}^{d}s(k_i)=\sum_{i=1}^{d}(6k_i-4) = -4d+6\sum_{i=1}^{d}k_i \le 6nd$ where $s(k_i)\le 6k_i-4$ is the circuit depth to perform an $SC$ operation and edge reset on $V_i$ (from \ref{sc_analysis}).}

Given a particular $\mathcal{SC}_{sys}$, we can also optimize the ordering. Note that a complete edge reset is not always necessary; some edges can potentially be preserved, thereby reducing both the number of required gates and the total number of time steps.

Another useful result \cite{alekseev2001orthogonal} is as follows.
\begin{theorem}
For any graph $G$ with $n$ vertices, \[\mathbb{c}_2(G) \le n-1 \,.\]  
\end{theorem}
If $c_2(G)$ is constant with respect to some class of graphs, then the graph states corresponding to that class can be efficiently distributed in a linear number of time steps since $t(\overline{k}) \le n$  so $T(G,\mathcal{SC}_{sys}) = O(n)$ (from \ref{GSDtime}). It turns out that we can characterize the class of graphs with a constant bounded value of $c_2(G)$ \cite{Buchanan2021SubgraphComp}. Dense structured graphs often have a bounded $\mathbb{c}_2(G)$, such as cliques ($\mathbb{c}_2(G)=1$) and complete bipartite graphs $(\mathbb{c}_2(G)=3)$. Specifically, the minimum distribution time $T(G,\mathcal{SC}_{sys})$ for a clique is achieved with $\mathcal{SC}_{sys} = \{V\}$. Similarly, for complete bipartite graphs with vertex partitions $P_1, P_2$ the corresponding optimal subgraph complementation system is given by $\mathcal{SC}_{sys} = \{P_1 \cup P_2,P_1,P_2\}$.  

It is worth noting that $\mathbb{c}_2(G)$ is related to another graph attribute -- the \textit{minimum rank} \cite{Buchanan2021SubgraphComp} of $G$ over $\mathbb{F}_2$.  Let $A$ denote the adjacency matrix of a graph $G$. The minimum rank over $\mathbb{F}_2$, $mr(G,\mathbb{F}_2)$, is defined as $mr(G,\mathbb{F}_2) = min_D$ $rank(A+D)$ over all binary diagonal matrices $D$. It is known that $\mathbb{c}_2(G) = mr(G,\mathbb{F}_2)$ or $\mathbb{c}_2(G) = mr(G,\mathbb{F}_2) + 1 \cite{Buchanan2021SubgraphComp}$. 

Calculating the minimum number of time steps taken to distribute a graph state and the corresponding $\mathcal{SC}_{sys}$ appears to be a hard problem in general, but we observe some common characteristics of entanglement distribution using subgraph complementations. There is often a trade-off between $t(\overline{k})$ and $\mathbb{c}_2(G)$, for some $\mathcal{SC}_{sys}$, such that their product $T(G,\mathcal{SC}_{sys}) = O(n)$ for structured or common classes of graphs as we show in Table 1. Intuitively, if larger subgraphs are complemented at each \added{SC operation} step, then fewer total steps are needed to build an arbitrary structured graph, and vice-versa. \added{From Theorem IV.1, we define $SC_{min}^G$ to be the minimal size $SC_{sys}$ corresponding to $G$, i.e. $\mathbb{c}_2(G)$, and claim that it describes a near optimal sequence of subgraph complementations to distribute $\ket{G}$. We note that greedy algorithms based on vertex covers provide good constructive constant factor approximations for $SC_{min}^G$ \cite{Buchanan2021SubgraphComp}, and hence also for minimizing $T(G,\mathcal{SC}_{sys})$}. Table \ref{cost} lists the asymptotic costs for several relevant classes of graphs. {Note that sparse graphs (graphs with $O(n)$ edges) always have $T(G,\mathcal{SC}_{sys})=O(n)$ since each edge can be added one at a time. Whereas this is not provably true for dense graphs (graphs with $\omega(n)$ edges) in general, we list some relevant examples where we also obtain a linear bound.}

\begin{table}[h!]
\centering
\caption{Costs for distributing various graph states}
\begin{tabular}{| c | c | c | c |}
    \hline
    \bf{Class} & $\mathbf{\mathbb{c}_2(G)}$ & $\mathbf{t(\overline{k})}$  & $\mathbf{T(G,\mathcal{SC}_{sys})}$ \\
    \hline
    \rowcolor{lightgray}
    \multicolumn{1}{|l|}{Sparse Graphs:} & $O(n)$ & $O(1)$ & $O(n)$ \\
    \hline
    Path Graph & $n-1$ & $O(1)$ & $O(n)$\\
    \hline
    Cycle Graph & $n-2$ & $O(1)$ & $O(n)$\\
    \hline
    Tree Graph & $\le n-2$ & $O(1)$ & $O(n)$\\
    \hline
    \rowcolor{lightgray}
    \multicolumn{1}{|l|}{Dense Graphs:} & ... & $O(n)$ & ...  \\ 
    \hline
    Complete Graph ($n$-GHZ state) & $1$ & $O(n)$ & $O(n)$ \\
    \hline
    Repeater Graph \cite{azuma2015repeater} & $\le 3+n$ & $O(1)$ & $O(n)$\\
    \hline
    Complete $m$-partite Graph & $m+1$ & $O(n)$ & $O(n)$ \\
    \hline
\end{tabular}
\label{cost}
\end{table}

\subsection{Entanglement of the graph state}
An interesting property of $\mathbb{c}_2(G)$ is its relation to the Schmidt-rank $r_G(X)$, a common entanglement measure of a pure state. For a graph state $\bk{G}$, $r_G(X)$ is defined over some bipartition $(X,V(G)\setminus X)$ of the vertices of $G$. For a pure state, the Schmidt-rank counts the minimal number of superpositions of factorizable states needed to generate the state \cite{sperling2011schmidt} -- a higher Schmidt-rank of a state implies it has higher entanglement. Define the average Schmidt-rank $\mathbb{E}_r(G)$ to be the expectation of $r_G$ for a uniformly chosen random subset of $V$.
\begin{theorem}
For any graph $G$, $\mathbb{c}_2(G) > \mathbb{E}_r(G)\,. $   
\end{theorem}
\begin{proof}
This can be inferred from a relation between $\mathbb{c}_2(G)$ and the cut-rank of a graph \cite{nguyen2020average}, which equals the Schmidt-rank \cite{dahlberg2018transforming}, as described below.
Define the cut-rank $\rho_G(X)$ of a set $X$ of vertices in a graph $G$ as the rank of the $X \times (V(G)\setminus X)$ matrix over the binary field whose $(i, j)$-entry is 1 if vertex $i$ in $X$ is adjacent to vertex $j$ in $V(G)\setminus X$ and 0 otherwise. 
Define the expected cut-rank $\mathbb{E}\rho(G)$ as the cut-rank for a uniformly chosen random subset $X$ of $V(G)$ \cite{nguyen2020average}. 
Now, consider the following two results:
\begin{align}\mathbb{c}_2(G) & > \mathbb{E}\rho(G) \,, \label{eq:c2erho} \\
\rho_G(X) & = r_G(X) \,, \label{eq:rhog-rg}
\end{align}
where (\ref{eq:c2erho}) comes from \cite{nguyen2020average} and (\ref{eq:rhog-rg}) from \cite{dahlberg2018transforming}. From these, we can directly conclude that $\mathbb{E}\rho(G) = \mathbb{E}_r(G)$, and thus, \[\mathbb{c}_2(G) > \mathbb{E}_r(G)\,.   \]
\end{proof}
Consequently, graph states with a high degree of entanglement will have large values of the subgraph complementation number $\mathbb{c}_2(G)$, and therefore likely to be more expensive to distribute.

\section{Detailed comparison to prior work}
Our approach utilizes $(n-1)$ EPR pairs, $O(\overline{k}n)$ bits of classical communication (CC), and a maximum of $n-1$ qubits at the central node, all of which improve on prior work \cite{meignant_markham_grosshans_2019,fischer2021distributing,avis_rozpedek_wehner_2023,cuquet2012growth}. Our approach requires one time step to distribute an arbitrary graph state, if we follow early work to implement our cost model \cite{meignant_markham_grosshans_2019,browne2011computational}, which make the assumption that local gates at multiple nodes and the same Clifford operations on a single central node commute, so all the operations can be performed in one time step. We instead assume a gate-based model of computation, which implies $T(G,\mathcal{SC}_{sys})$ number of time steps is required. Other work \cite{avis_rozpedek_wehner_2023,cuquet2012growth} utilizes an approach with a central network node that takes $O(n)$ time steps if local operations are free and $O(|E|) = O(n^2)$ time steps if not (accounting for local gates) to build a graph state corresponding to a dense graph. In the presence of memory decoherence, local operations are not free. The fidelity of the state directly depends on the number of time steps required. Our approach also utilizes $O(n^2)$ number of gates and time steps to distribute certain dense graph states. However, many common graph states, e.g., GHZ states, utilize a number of gates and time steps that is $O(n)$ since $T(G,\mathcal{SC}_{sys}) = O(n)$ for those graph states. This presents a quadratic improvement in the number of time steps, gates required, and potential fidelity of the state in the presence of noise. Table \ref{cost} includes the various costs of different classes of graphs.

\section{Incorporating Parallelization}
We first describe a log depth circuit to generate a $n$ qubit GHZ state, that is the same up to local rotations as the construction proposed by Cruz et al. \cite{cruz2019efficient}. This will be useful as a subroutine in reducing the time to create a star graph, a graph consisting of one central vertex connected with edges to all the other vertices,  as in steps 1 and 3 of the subgraph complementation operation from \ref{SCsingle}.

\subsection{Log depth GHZ state generation}
\label{log_depth}
Denote the $n$ qubits as $q_0,q_1, q_2... q_{n-1}$. Begin with a GHZ state (star graph state with center $q_0$) of a constant size of $c$ qubits s.t. $c<n/2$ which can be created in constant depth through the application of $c$ $CZ$ gates. We now present the procedure to double the size of the GHZ state in constant depth.

\noindent\rule{\linewidth}{2pt}

\noindent\textbf{GHZ state doubling algorithm}

\vspace*{-0.7em}
\noindent\rule{\linewidth}{1pt}

\begin{enumerate}[1.]
  \label{alg:ghzdoub}
  \item Consider a GHZ state $\ket{G}$ on the first $c$ qubits $q_0, q_1, \dots, q_c$ where $c<n/2$.
  \item For $i = 1$ to $c$ \textit{in parallel}, apply $CZ(q_i, q_{c+i})$.
  \item For $i = 1$ to $c$ \textit{in parallel}, apply $LC(q_i)$.
  \item For $i = 1$ to $c$ \textit{in parallel}, apply $CZ(q_i, q_{c+i})$.
  \item Now, $\ket{G}$ is a GHZ state on $2c$ qubits.
\end{enumerate}

\noindent\rule{\linewidth}{2pt}

\noindent
\textbf{Time Analysis:} Each step in the GHZ state doubling algorithm consists of $c$ operations that can be performed in parallel. Thus, each for-loop contributes only one time step, and the entire procedure takes constant depth (three time steps in total).
The above algorithm is equivalent to applying $CZ(q_0,q_{c+i}), \forall q_i: 1 \le i \le c$ but only requires three time steps (instead of $c$ time steps). Thus, the size of the GHZ state is doubled from $c$ to $2c$ in constant time. Repeated application of this algorithm clearly leads to the creation of an $n$-qubit GHZ state in $O(\log(n))$ depth. 

\subsection{Application to subgraph complementation}
The above algorithm can be utilized to reduce the times for steps 1 and 3 of the subgraph complementation operation (from \ref{SC}). We first observe that both steps induce a star graph (with the edges added modulo 2) on qubits $a_0,\ldots ,a_{k-1}$ with $a_0$ as the center.
\label{parSCalg}
Let $n$ qubits reside in the central node. Naively, $n-1$ $CZ$ gates (and time steps) are required to construct the star graph. We claim that by introducing $k$ auxiliary qubits $q_0, q_1... q_{k-1}$, $k = 1,2,\ldots ,n$, in the central node, the star graph can be generated in $O(n/k + \log(k))$ time steps. Since steps 2 and 4 of the subgraph complementation operation both take one time step, this reduces the asymptotic runtime for carrying out an $SC$ operation to the time required to add a star graph (in steps 1 and 3). The parallelized \textit{SC} operation is described as follows.
\begin{figure}[H]
  \centering
  \includegraphics[scale=0.3]{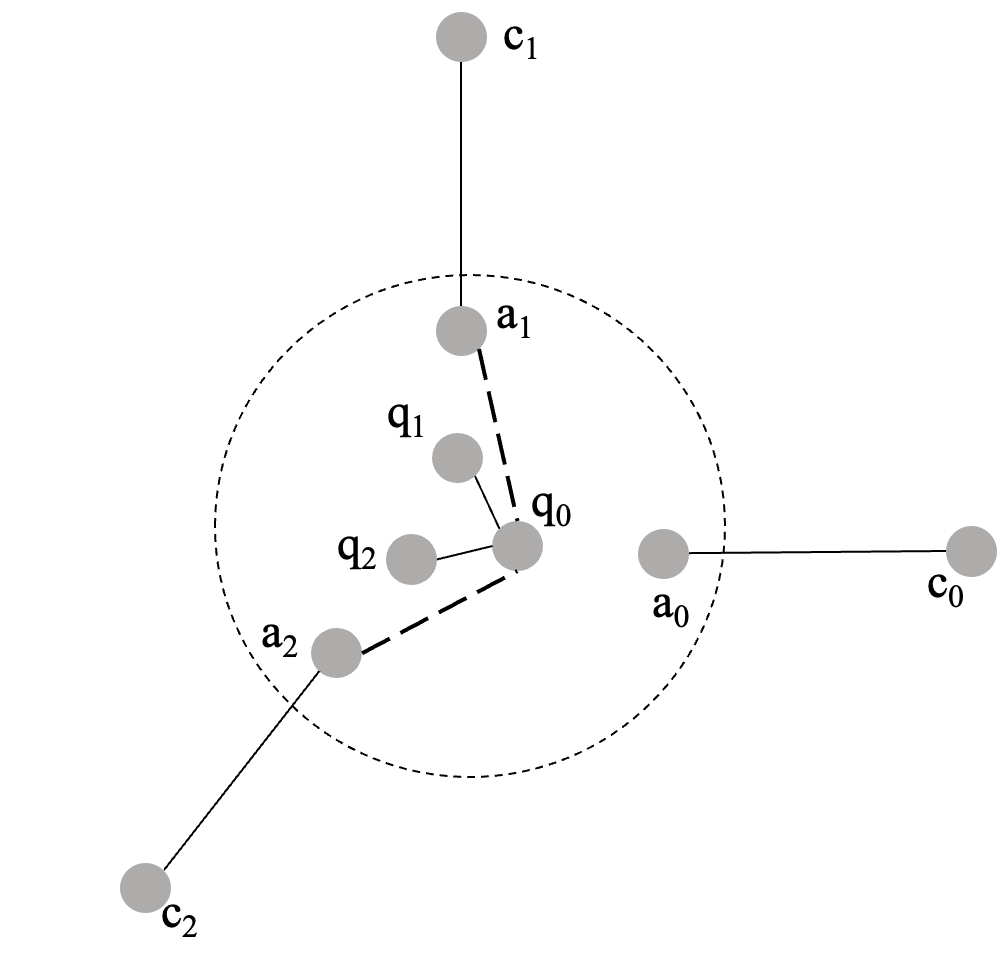}
  \caption{Parallelized subgraph complementation. The solid lines indicate the inner edges whereas the dashed lines indicate the outer edges which were added in step (c) of the parallelized $SC$ operation.}
  \label{fig:parSC}
\end{figure}
\vspace*{-1.5em}
\noindent\rule{\linewidth}{2pt}

\noindent\textbf{Parallelized \textit{SC} operation for steps 1 and 3}

\vspace*{-0.7em}
\noindent\rule{\linewidth}{1pt}

\begin{enumerate}[(a)]
    \item Consider a network with a central node $A$ that contains $n$ primary qubits $a_0,...,a_{n-1}$ and $k$ auxiliary qubits $q_0,...,q_{k-1}$. The network also contains $n-1$ nodes $C_0,...,C_{n-2}$ each containing a qubit $c_0,...,c_{n-2}$.
    \item Create a $k$-qubit GHZ state at the central node on the auxiliary qubits $q_0,...,q_{k-1}$ by repeating  the GHZ state doubling algorithm $\log(k)$ times (from \ref{alg:ghzdoub}).
    \item Add $(k-1)$ edges of the form $(q_0,a_i)$ through the GHZ state doubling algorithm (from \ref{alg:ghzdoub}) . This can be done on any subset of size $(k-1)$ of qubits $a_i$.
    \item Repeat step (b) $n/(k-1)$ times to create a star graph with $(k-1)$ \textit{inner} edges (between $q_0$ and $q_i:i\ne0$, from the initial $k$-qubit GHZ state) and $n$ \textit{outer} edges (between $q_0$ and $a_i$). Fig. \ref{fig:parSC} illustrates this construction.
    \item Remove the $(k-1)$ inner edges, similar to step (a).
\end{enumerate}

\noindent\rule{\linewidth}{2pt}

Note that step 4 of the original subgraph complementation operation (from \ref{SC}) now becomes $LC(q_0)$ instead of $LC(a_0)$.

\subsection{Analysis}
Let $T(G)$ denote the time required by the parallelized SC algorithm to distribute a graph state $\bk{G}$. Denote by $t(n)$ the time required for a single subgraph complementation on $n$ vertices, $n \in \mathbb{N}$. Steps (b) and (e) of the above operation take $O(\log n)$ time. Steps (c) and (d) take $O(n/k)$ time. The other steps in the original subgraph complementation operation (from \ref{SCsingle}) take constant time. Thus, with $(n+k)$ qubits in the central node, 
\[t(n) = O\left(\frac{n}{k} + \log(k)\right).\] 
Choosing $k = n/\log(n),$
\[t(n) = O(\log(n)).\]
Consequently, 
\begin{equation}
T(G) = \mathbb{c}_2(G) \cdot t(\overline{n}) = O(n\log(n)) \, .    
\end{equation}
\\
\noindent {This derivation has the simple corollary that is if $\mathbb{c}_2(G)=m$ for some constant $m$, then $T(G) = O(\log(n))$. In particular, this provides an exponential improvement for generating classes of graph states corresponding to such graphs when compared to direct parallelization of applying CZ gates, which is the best known method generally, and which requires $O(n)$ circuit depth \cite{cabello_danielsen_lópez-tarrida_portillo_2011}. An example of such a class of graphs is a complete $m-$partite graph, which has $\mathbb{c}_2(G)=m$ for some constant $m$. We discuss this further in section \ref{par}}.
\\

\section{Graph State Generation in the Presence of Noise} \label{sec:noisy}
In this section, we present numerical results comparing our graph state distribution algorithm to the algorithms proposed by Avis et al. \cite{avis_rozpedek_wehner_2023} and Cuquet, Calsamiglia \cite{cuquet2012growth}, which are closest to ours in terms of problem setting and resource consumption\footnote{All the code used to carry out simulations, including programs to create the graphs in this section, are publicly available at \href{https://github.com/Anisen123/graph-state-simulator}{https://github.com/Anisen123/graph-state-simulator}.}.

To carry out our simulations, we model noise through depolarizing channels. The action of a depolarizing channel on a state $\rho$ can be expressed by  \cite{nielsen2001quantum}:

\begin{equation}
\mathcal{E}(\rho) \xrightarrow{} (1-p)\rho + \frac{p}{3}(X\rho X + Y\rho Y + Z\rho Z)\, .  
\label{eq: dep_noise}
\end{equation}

This is equivalent to applying $X,Y$ or $Z$ gates on $\rho$ with equal probability $p/3$ and the identity $I$ with probability $(1-p)$. Each time an operation ($LC$ or $CZ$) is carried out on some qubit(s) $q_i$ with state $\rho_i$, we model noise by applying \ref{eq: dep_noise} to each qubit. Thus, the graph state $\bk{G}$ is transformed as follows:
\[\bk{G} \xrightarrow{} P_{i}\bk{G} \,,\]
where the Pauli operator $P_{i} \in \{X,Y,Z,I\}$  acts on the state $\rho_i$ and $P_{i}$ is selected according to probabilities described above. We can also express $P_{i}\bk{G}$ as
\[P_{i}\bk{G} = Z^b\bk{G}\,,\]
for some bitstring $b=b_1b_2...b_n$, where $Z$ is applied on the $i$'th qubit iff $b_i=1$, since the action of any Pauli operator \added{$P_i$} on $q_i$ can be mapped to $Z$ gates acting on some set of qubits in the graph state \cite{griffiths1979graph},\added{ which can include $q_i$ and its neighboring qubits. In particular, single qubit $X,Y,Z$ errors commute through $CZ$ gates in $\ket{G}$, after which the $X,Y$ errors can be converted into $I,Z$ errors respectively (since $X\ket{+}=\ket{+}$)}.  Thus, we can now express all noise acting on $\bk{G}$ as a single bitstring of $Z$ operators acting on $\bk{G}$ since $Z$ gates commute. This makes the noise easier to characterize and store in a list data structure. We estimate the fidelity by running this simulation several times (between 5000 and 10000 times) and finding the fraction of states where $b$ is the all zero bitstring. Note that the probability that $b$ is the all zero bitstring corresponds to the fidelity of the total mixed state. This is because the set of all possible states of the form $Z^b\bk{G}$ corresponding to $2^n$ possible linearly independent $Z$-bitstrings are all orthogonal to each other.

The red line in each simulation graph indicates an arbitrarily chosen fidelity of 0.75, which serves as a reference line or a visual aid to compare fidelities between different graphs.

\subsection{GHZ State Generation}
\textbf{Noise Model.} We consider noise incurred from carrying out $CZ,LC$ gate operations, measuring qubits in the $Y$ or $Z$ basis, and local Clifford correction unitaries. We assume Bell pairs that are initially noiseless. While noisy Bell pairs will impact the final fidelity, we do not expect the relative performance of the algorithms to change, since each algorithm uses roughly the same initial number of Bell pairs (in fact, our algorithm uses one less Bell pair). In Appendix \ref{appendix}, we provide simulation results that consider memory decoherence as well in Figs. \ref{fig:sim_mem1},\ref{fig:sim_mem2}. First we compare the fidelity obtained using our algorithm (from \ref{SCalg}) with that obtained using the Factory Node + Teleportation (FNT) Algorithm (from \ref{fntalg}) (or the Bipartite A protocol) and the Bipartite B protocol \cite{cuquet2012growth} to generate a 3-qubit GHZ state. Fig. \ref{fig:plot3q_comp_all} (left) compares the fidelity as a function of probability $p$ of applying a non-identity Pauli gate at each operation. Note that we implement the local creation of a GHZ state (for Bipartite A and B) by creating a star graph state first and then applying an $LC$ operation on the central node.

\begin{figure}[h]
  \begin{subfigure}{0.49\textwidth}
  \includegraphics[scale=0.5]{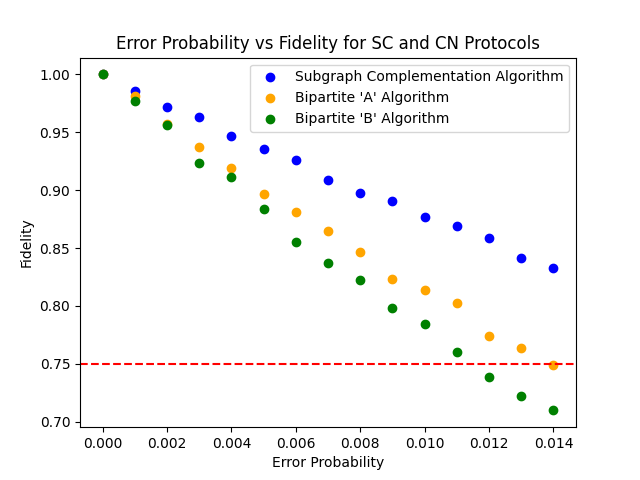}
  \end{subfigure}
  \begin{subfigure}{0.49\textwidth}
  \includegraphics[scale=0.5]{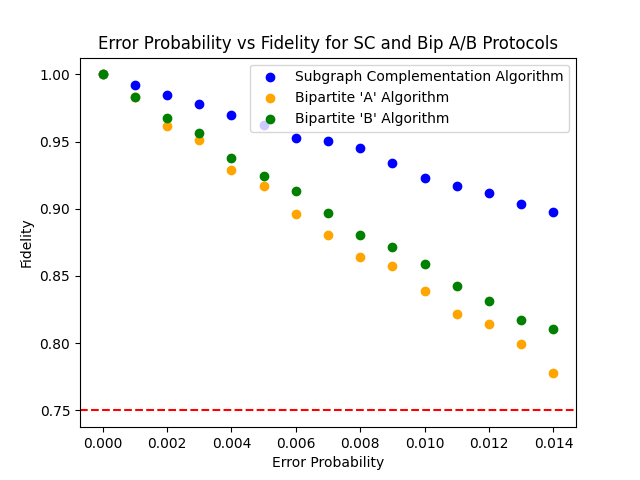}
  \end{subfigure}
    \caption{Comparison of fidelity of a 3-qubit GHZ state distributed according to the SC algorithm (\ref{SCalg}), the FNT algorithm (Bipartite A protocol) (\ref{fntalg}), and the Bipartite B protocol, considering (left) and neglecting (right) noise from local correction unitaries, having $0\le p\le 0.015$, with the results obtained over 10000 trials. The red line indicates a fidelity of 0.75.}
    \label{fig:plot3q_comp_all}
\end{figure}

We observe that our algorithm generates higher fidelities than the others. We also note that the Bipartite A protocol generates GHZ states with higher fidelity compared to the Bipartite B protocol, despite the latter using fewer gate operations \cite{cuquet2012growth}. Note that this differs from prior results reported by Cuquet and Calsamiglia \cite{cuquet2012growth}. However, we recover their results in Fig. \ref{fig:plot3q_comp_all} (right), when we neglect the noise from local correction unitaries, as was done in their paper but which we do not. We support these results by showing that the same trend holds when generating graph states with a varying number of qubits, Fig. \ref{fig:plotnq_comp_all} in Appendix \ref{appendix}. Consequently, to simplify analysis, we compare our algorithm only with the FNT Algorithm (\ref{fntalg}) (or Bipartite A) for the rest of the paper since it experimentally performs better following our noise model.

Fig. \ref{fig:sim_multqubits} shows fidelity vs. number of qubits for a fixed $p=10^{-4}$ for both the FNT and Subgraph Complementation (SC) Algorithms. We observe little or no difference between these two algorithms when the number of qubits is less than or equal to $x$, and that the SC algorithm produces higher fidelity states otherwise and that the difference increases with number of qubits.  

\begin{figure}
    \centering
    \includegraphics[scale=0.60, trim={0.43cm 0 1.2cm 1.4cm},clip]{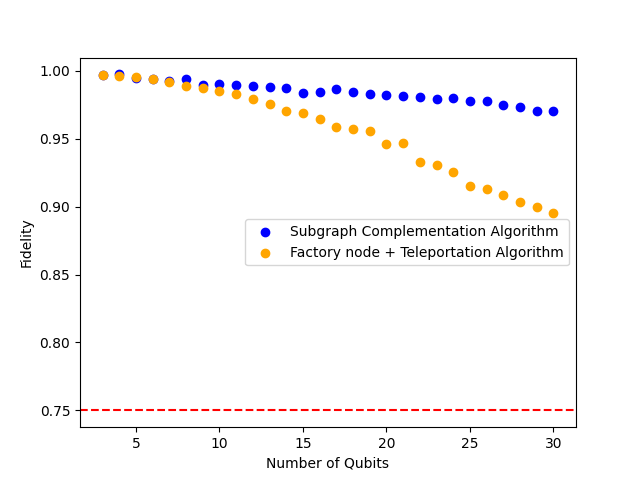}
    \caption{Comparison of fidelity of an $n$-qubit GHZ graph state distributed according to the SC algorithm (\ref{SCalg}) and the FNT algorithm \cite{avis_rozpedek_wehner_2023} varying $6\le n\le 30$ with $p = 10^{-4}$, with the results obtained over 5000 trials. The red line indicates a fidelity of 0.75.}
    \label{fig:sim_multqubits}
\end{figure}

\subsection{Complete bipartite graph state generation}
Complete bipartite graphs have applications in entanglement purification \cite{yan2023advances}, silicon-photonic quantum computing \cite{rudolph2017optimistic} and are also used as graph states for all-photonic repeaters \cite{azuma2015repeater}. Here, we consider complete bipartite graphs with equal size partitions. Fig. \ref{fig:sim_bip1} and Fig. \ref{fig:sim_bip2} compare the fidelity obtained by our algorithm with that of \added{the FNT algorithm (\ref{fntalg}) \cite{avis_rozpedek_wehner_2023,cuquet2012growth}}, for complete bipartite graph states with 6 qubits and 20 qubits respectively. Our algorithm creates complete bipartite graph states with lower fidelity when the number of qubits is small (less than 12 qubits), but with higher fidelity for larger states with 20 qubits. Fig. \ref{fig:sim_bip_mult} compares fidelity when the number of qubits is varied between 6 and 30 (with steps of 2 qubits), for error probability $p = 10^{-3}$. 

\begin{figure}[h]
  \centering
  \begin{subfigure}[h]{0.49\textwidth}
      \includegraphics[scale=0.5]{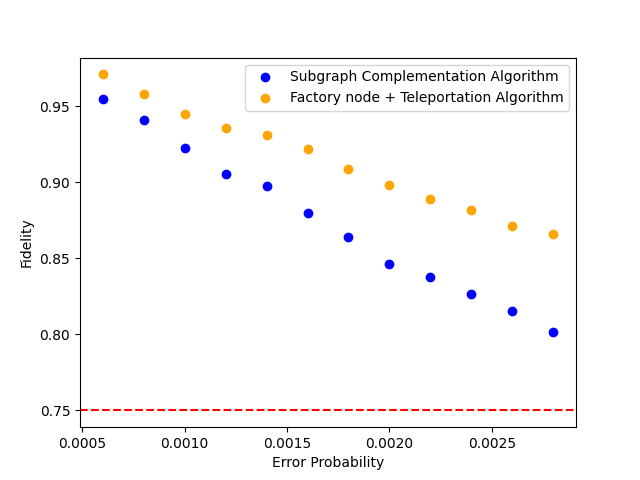}
      \caption{6-qubit complete bipartite graph state}
      \label{fig:sim_bip1}
  \end{subfigure}
  \hfill
  \begin{subfigure}[h]{0.49\textwidth}
      \includegraphics[scale=0.5]{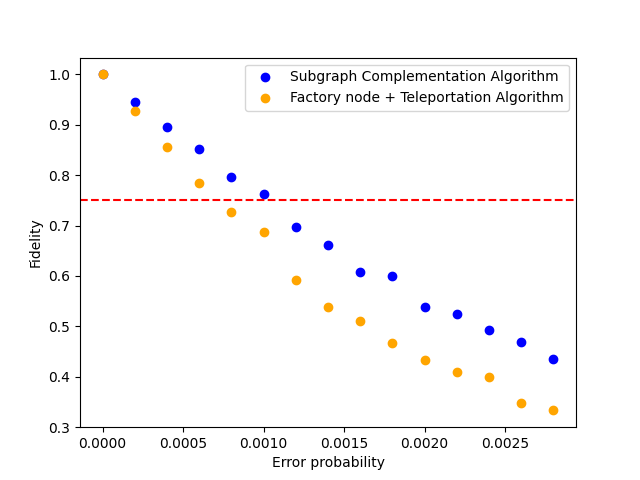}
      \caption{20-qubit complete bipartite graph state}
      \label{fig:sim_bip2}
  \end{subfigure}
  \caption{Comparison of fidelity of \textbf{(a)} 6-qubit complete bipartite graph state and \textbf{(b)} 20-qubit complete bipartite graph state, distributed according to the SC algorithm (\ref{SCalg}) and the FNT algorithm \cite{avis_rozpedek_wehner_2023} having $0\le p\le 0.003$, with the results obtained over 10000 trials for each of the graph states. The red line indicates a fidelity of 0.75.}
  \label{bip_graphs}
\end{figure}

We expect the difference in fidelity to monotonically increase in favor of the SC algorithm as the number of qubits increase. We conjecture this difference is a consequence of the $O(n)$ scaling in time and gate operations\added{, or circuit depth,} for our algorithm as compared to $O(n^2)$ for FNT (\ref{fntalg}). For a smaller number of qubits, a linear overhead due to the edge reset operation results in a higher number of operations to carry out a subgraph complementation system as compared to locally creating a graph state and teleporting all the qubits. However, as the number of qubits increases asymptotically, the fidelity to distribute a complete bipartite graph state for our algorithm should increase only linearly as compared to quadratically for FNT (\ref{fntalg}). We observe that Fig. \ref{fig:sim_bip_mult} exhibits this trend \added{with the crossover in fidelities observed at $n=12$}. Note that we set $p=10^{-3}$ in the figure as compared to $p=10^{-4}$ in Fig. \ref{fig:sim_multqubits} since at $p=10^{-4}$ complete bipartite graph states (as opposed to GHZ states) created by both algorithms are less noisy and too close in fidelity for a meaningful comparison. This is our guiding principle for choosing parameters in all figures, with comparisons expected to generalize in other parameter regimes.

\begin{figure}[h]
  \centering
  \includegraphics[scale=0.60, trim={0.43cm 0 1.2cm 1.4cm},clip]{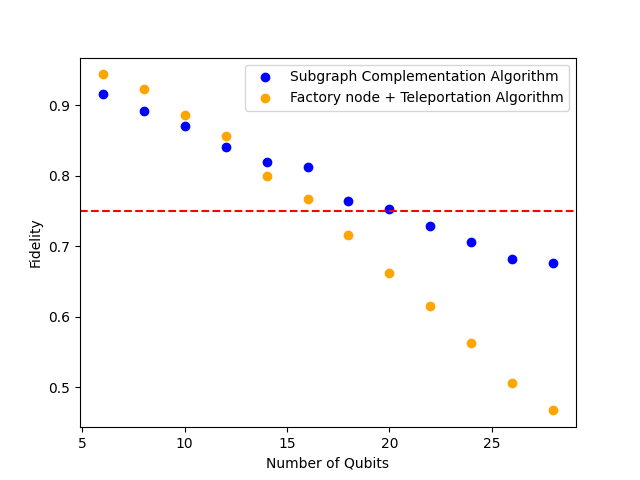}
  \caption{Comparison of fidelity of an $n$-qubit complete bipartite graph state distributed according to the SC algorithm (\ref{SCalg}) and the FNT algorithm \cite{avis_rozpedek_wehner_2023} varying $6\le n\le 30$ with $p = 10^{-3}$, with the results obtained over 5000 trials. The red line indicates a fidelity of 0.75.}
  \label{fig:sim_bip_mult}
\end{figure}

\subsection{Parallelized GHZ state generation}\label{par}
We also carried out simulations of distributing a GHZ state with parallelized gate operations as described in section VI. We compare our parallelized SC algorithm (\ref{parSCalg}) to a parallelized version of the FNT algorithm (\ref{fntalg}), in which a complete graph is created locally in the central node by parallelizing $CZ$ gates according to an edge coloring of the complete graph \cite{cabello_danielsen_lópez-tarrida_portillo_2011}. This algorithm requires $O(n)$ time steps to distribute an $n$-qubit GHZ state, however the parallelized subgraph complementation algorithm only requires $O(\log n)$ time steps. Fig. \ref{fig:sim_par_small} compares the fidelity produced by our parallelized SC algorithm with that produced by the parallelized FNT algorithm for a 6-qubit GHZ state, as a function of error probability. The parallelized SC algorithm uses 2 auxiliary qubits. 

\begin{figure}[h]
  \centering
  \centering
  \begin{subfigure}[h]{0.49\textwidth}
      \centering
      \includegraphics[scale=0.5]{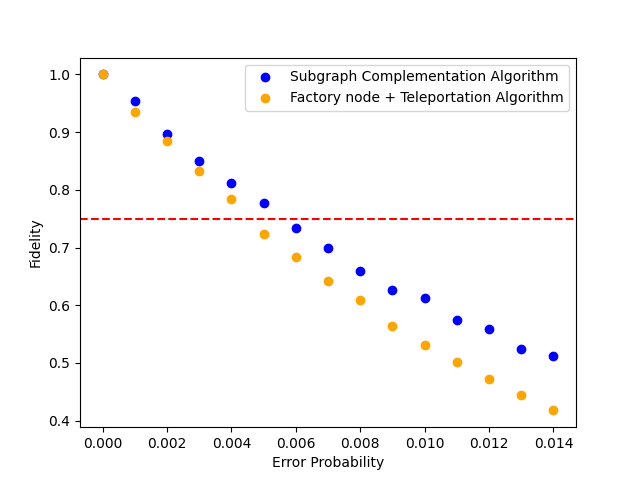}
      \caption{6-qubit GHZ state (Parallel Implementation)}
      \label{fig:sim_par_small}
  \end{subfigure}
  \hfill
  \begin{subfigure}[h]{0.49\textwidth}
      \centering
      \includegraphics[scale=0.5]{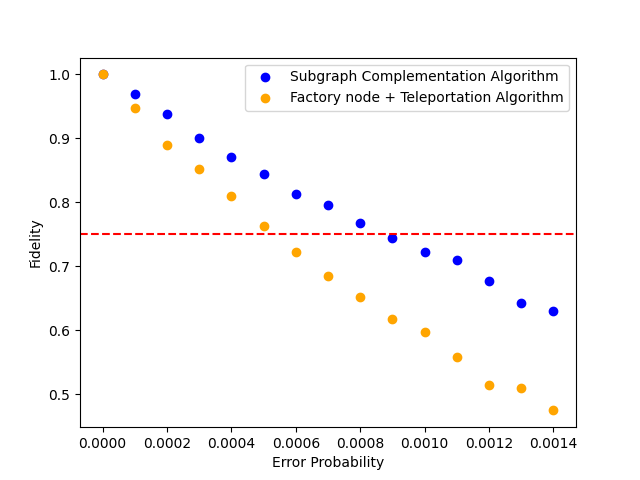}
      \caption{20-qubit GHZ state (Parallel Implementation)}
      \label{fig:sim_par_big}
  \end{subfigure}
  \caption{Comparison of fidelity of \textbf{(a)} 6-qubit GHZ state and \textbf{(b)} 20-qubit GHZ state, distributed according to the parallelized SC algorithm (with \textbf{(a)} 2 auxiliary qubits and \textbf{(b)} 4 auxiliary qubits respectively) and the FNT algorithm having $0\le p\le 0.015$, with the results obtained over 10000 trials. The red line indicates a fidelity of 0.75.}
  \label{fig:par_graphs}
\end{figure}

Our algorithm provides marginally higher fidelity in creating a 6-qubit GHZ state, however we observe in Fig. \ref{fig:sim_par_big} that this difference increases for distributing a larger GHZ state with 20 qubits having error probability of the order of $10^{-4}$.

This trend clearly shows up in Fig. \ref{fig:sim_par_multqubits}, where the number of qubits varies between 2 and 30 to create the corresponding GHZ state. 

\begin{figure}[h]
  \centering
  \includegraphics[scale=0.57, trim={0.43cm 0.2cm 1.2cm 1.4cm},clip]{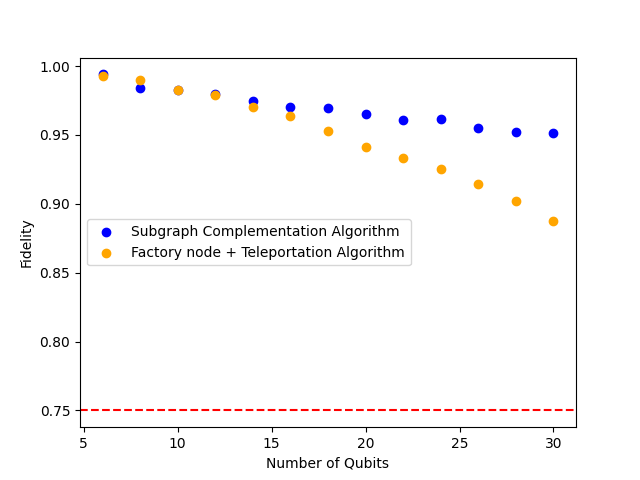}
  \caption{Comparison of fidelity of an $n$-qubit GHZ graph state distributed according to the parallelized SC algorithm and the FNT algorithm varying $6\le n\le 30$ with $p = 10^{-4}$, with the results obtained over 5000 trials. The red line indicates a fidelity of 0.75.}
  \label{fig:sim_par_multqubits}
\end{figure}

It is interesting to note that on comparing Fig. \ref{fig:sim_multqubits} and Fig. \ref{fig:sim_par_multqubits}, we observe that the fidelity obtained from distributing $n$-qubit GHZ states without parallelization is higher than with parallelization (although both outperform the parallelized FNT algorithm). This is expected since we are using $n/\log(n)$ extra qubits for the parallelized algorithm. However, we are also gaining an exponential decrease in distribution time ($O(\log n)$ with parallelization as compared to $O(n)$ without it). Therefore, this tradeoff would be favourable in most cases, especially when considering memory decoherence of the qubits which would increase proportionally with distribution time.

\section{Conclusion}
We have presented a novel and efficient algorithm for distributing an arbitrary graph state across a quantum network by drawing upon the mathematical theory of \textit{subgraph complementations}, which we show to be strongly linked to the entanglement and structure of the graph state. We show, through simulations, that this algorithm generates graph states having high fidelity in the presence of noisy operations. This algorithm also provides a natural framework to create a graph state in a distributed fashion, since subgraph complementations commute and it may be possible to implement several $SC$ operations in parallel (since their order does not matter) and merge the subgraphs at the end.

This could also be extended to a procedure for distributing arbitrary stabilizer states, since graph states are equivalent to stabilizer states up to local operations. We anticipate that the framework of subgraph complementations will be useful for other applications as well, for instance quantum error correction and secret sharing. Further work includes incorporating error correction in estimating runtime and fidelity through noisy simulations of our algorithm, considering the detailed impact of memory decoherence, as well as optimizing our algorithm for more domain-specific purposes.

\section{Acknowledgements}
This work was funded by the Army Research Office
(ARO) MURI on Quantum Network Science under grant
number W911NF2110325, NSF-ERC Center for Quantum Networks grant EEC-194158 and NSF under grant  2402861.

\bibliographystyle{plainnat}
\bibliography{references}

\appendix
\section{Additional simulation results}
\label{appendix}

\begin{figure}[h]
  \centering
  \begin{subfigure}{0.49\textwidth}
  \includegraphics[scale=0.50, trim={0.43cm 0 1.2cm 1.4cm},clip]{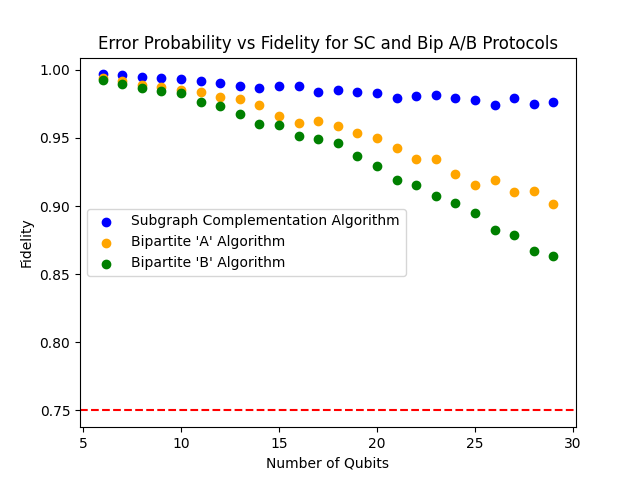}
  \end{subfigure}
  \begin{subfigure}{0.49\textwidth}
  \includegraphics[scale=0.50, trim={0.43cm 0 1.2cm 1.4cm},clip]{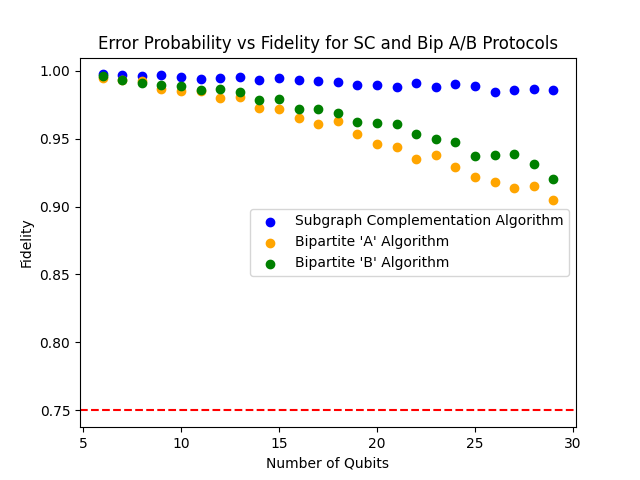}
  \end{subfigure}
    \caption{Comparison of fidelity of an $n$-qubit GHZ graph state distributed according to the SC algorithm (\ref{SCalg}), the FNT algorithm (Bipartite A protocol) (\ref{fntalg}), and the Bipartite B protocol, considering (top) and neglecting (bottom) noise from local correction unitaries, varying $6\le n\le 30$ with $p = 10^{-3}$, with the results obtained over 5000 trials. The red line indicates a fidelity of 0.75.}
    \label{fig:plotnq_comp_all}
\end{figure}

In this appendix, we provide some additional simulation results. Figure \ref{fig:plotnq_comp_all} shows the fidelity of $n$-qubit GHZ graph states generated using three different algorithms: the SC algorithm, the FNT algorithm (Bipartite A), and the Bipartite B protocol, for $6 \leq n \leq 30$ with error probability $p = 10^{-3}$. The top plot includes noise from local correction unitaries, while the bottom plot neglects this noise. We observe that the SC algorithm consistently achieves higher fidelity than the other algorithms across all $n$ values, even when local correction noise is accounted for. While the Bipartite B protocol requires fewer gate operations, the Bipartite A protocol outperforms it in fidelity, diverging from earlier results reported by Cuquet and Calsamiglia \cite{cuquet2012growth}. Interestingly, when local correction noise is excluded (bottom plot), we recover those earlier trends, showing improved performance of Bipartite B, consistent with assumptions made in prior work. These findings highlight the significance of accounting for realistic noise sources and support the robustness of the SC algorithm under more comprehensive noise models.

\begin{figure}[h]
  \centering
  \begin{subfigure}{0.49\textwidth}
  \includegraphics[scale=0.50, trim={0.43cm 0 1.2cm 1.4cm},clip]{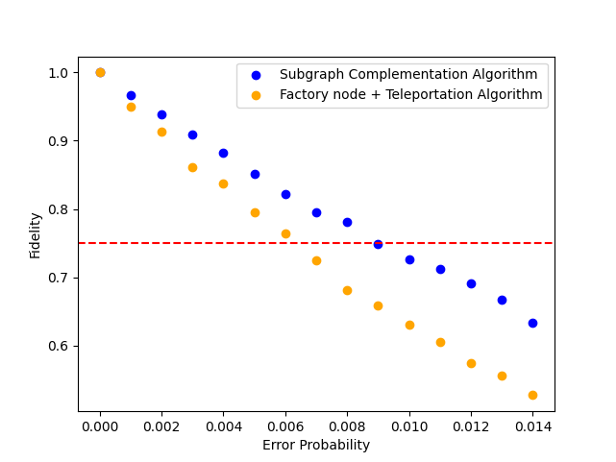}
  \end{subfigure}
  \begin{subfigure}{0.49\textwidth}
  \includegraphics[scale=0.50, trim={0.43cm 0 1.2cm 1.4cm},clip]{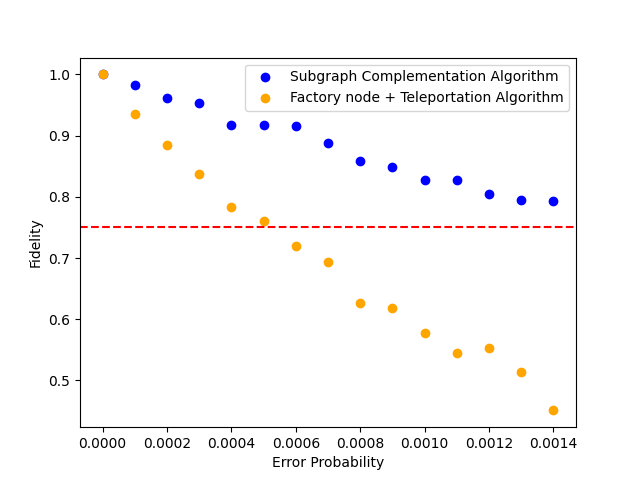}
  \end{subfigure}
  \caption{Comparison of fidelity of a 5-qubit (top) and 20-qubit (bottom) GHZ state distributed according to the SC algorithm (\ref{SCalg}) and the FNT algorithm \cite{avis_rozpedek_wehner_2023} having $0\le p\le 0.0015$, with the results obtained over 10000 trials. The red line indicates a fidelity of 0.75.}
  \label{fig:sim_big}
\end{figure}

Figure \ref{fig:sim_big} compares the fidelity of GHZ states of 5 qubits (top) and 20 qubits (bottom) generated using the SC and FNT algorithms as a function of the noise parameter $0 \leq p \leq 0.0015$. In both cases, the SC algorithm outperforms the FNT algorithm across the entire range of noise values. Notably, the gap in fidelity between the two algorithms becomes more pronounced in the 20-qubit scenario.

\begin{figure}[h]          
  \centering
  \includegraphics[scale=0.60, trim={0.43cm 0 1.2cm 1.4cm},clip]{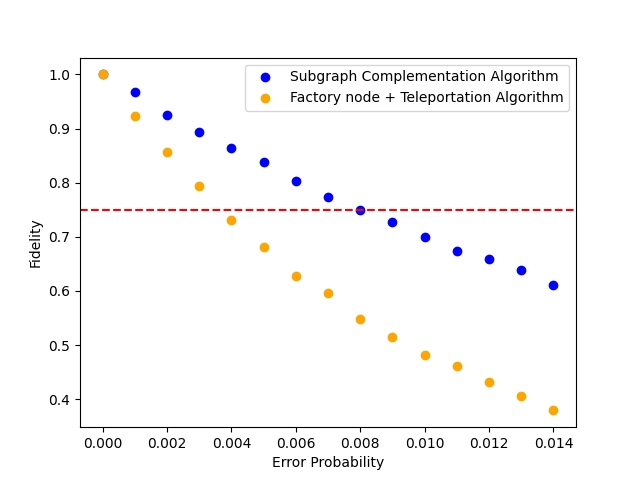}
  \caption{(Simulation considers noisy gate operations and memory decoherence) Comparison of fidelity of a 3-qubit GHZ state distributed according to the SC algorithm (\ref{SCalg}) and the FNT algorithm \cite{avis_rozpedek_wehner_2023} having $0\le p\le 0.0015$, with the results obtained over 10000 trials. The red line indicates a fidelity of 0.75.}
  \label{fig:sim_mem1} 
\end{figure}

Figures~\ref{fig:sim_mem1} and \ref{fig:sim_mem2} present the fidelity of GHZ states distributed using the SC and FNT algorithms under a more realistic noise model that includes both gate errors and memory decoherence. In Figure~\ref{fig:sim_mem1}, we consider a 3-qubit GHZ state with varying noise levels $0 \leq p \leq 0.0015$, while in Figure~\ref{fig:sim_mem2}, we examine $n$-qubit GHZ graph states for $6 \leq n \leq 30$ with a fixed error rate of $p = 10^{-4}$. In both scenarios, the SC algorithm consistently outperforms the FNT algorithm, though the difference does not seem to increase significantly with increasing number of qubits, unlike in other cases which did not consider memory decoherence. We conjecture this might be a consequence of a larger scaling being required with respect to the number of qubits in this case, since the noise dominates with a smaller number of qubits.

\begin{figure}[h]
  \centering
  \includegraphics[scale=0.60, trim={0.43cm 0 1.2cm 1.4cm},clip]{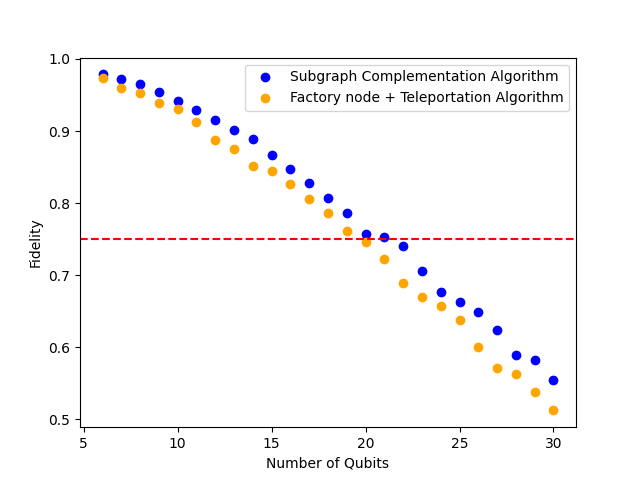}
  \caption{(Simulation considers noisy gate operations and memory decoherence) Comparison of fidelity of an $n$-qubit GHZ graph state distributed according to the SC algorithm (\ref{SCalg}) and the FNT algorithm \cite{avis_rozpedek_wehner_2023} varying $6\le n\le 30$ with $p = 10^{-4}$, with the results obtained over 5000 trials. The red line indicates a fidelity of 0.75.}
  \label{fig:sim_mem2}
\end{figure}

\end{document}